\documentclass[tikz,12pt]{article}
\usepackage{mheckII}

\usepackage[T1]{fontenc}

\usepackage{tikz-feynman}

\usepackage{manfnt}
\usepackage{longtable}

\usepackage{fancybox}

\usepackage{multirow}

\usepackage{blkarray}

\usepackage{tikz} 
\usetikzlibrary{matrix}

\theoremstyle{theorem}

\newtheorem{fact}{Fact}

\newtheorem{corl}{Corollary}
\newtheorem{pro}{Proposition}

\newtheorem*{conj}{Sub-conjecture}

\newtheorem{lem}{Lemma}
\newtheorem*{lem*}{Lemma}

\theoremstyle{definition}

\newtheorem*{condi}{Condition}

\newtheorem{rem}{Remark}

\newtheorem{exe}{Example}[section]

\def\Dsl{\,\raise.15ex\hbox{/}\mkern-13.5mu D}
\def\dsl{\,\raise.25ex\hbox{/}\mkern-10.5mu \partial}

\title{A Property of Geodesics\\ in Special K\"ahler Geometry 
}

\authors{Sergio Cecotti\footnote{e-mail: {\tt cecotti@bimsa.cn}}\vskip 9pt

\centerline{Beijing Institute of Mathematical Sciences and Applications (BIMSA)}
\centerline{Huaibei Town, Huairou District, Beijing 101408, China,}
\centerline{Qiuzhen College, Tsinghua University, Beijing, China,}
\centerline{SISSA,
via Bonomea 265, Trieste, Italy}
}

\abstract{We study the stable geodesics of the QFT special K\"ahler geometry ($\equiv$ Seiberg-Witten geometry of 4d $\cn=2$ QFT) using the Myers argument. Complete stable geodesics are quite restricted, and can be described very explicitly. In particular no closed stable geodesic exists. We comment on the application of the Myers method to related problems, including geodesics in moduli spaces of Calabi-Yau 3-folds.}

\begin{document}

\maketitle

\tableofcontents


\section{Introduction}

In this note we focus on the special K\"ahler geometry which describes the vector-multiplet sector of a 4d $\cn=2$ QFT, a.k.a.\! Seiberg-Witten geometry \cite{R1,R2,R3,R4,R5,R6,R7,R7.5,R8}. It should not be confused with the special K\"ahler geometry (sometimes called `projective')  which applies in presence of gravity, i.e.\! in $\cn=2$ SUGRA \cite{R5,R8,R9,R10,R11}, which is equivalent to the Weil-Petersson (WP) geometry on the moduli space of Calabi-Yau 3-folds \cite{R5,R8,R9,R10,R11,R12}.

A special K\"ahler manifold $\cm$ carries several geometric 
structures, and therefore can be studied from a variety of viewpoints.\footnote{\ For instance: as holomorphic integrable systems \cite{R7,R7.5}, or in terms of
variations of Hodge structure (VHS) \cite{R7,R10}, from an algebro-geometric viewpoint \cite{R13},
or using techniques from diophantine geometry \cite{R13,R14}, etc.}
The Riemannian aspects of these geometries -- such as their geodesics --
 have been less studied because, while all other 
structures are well behaved, the K\"ahler metric looks not too nice:  except for ``trivial'' situations,
it is \emph{singular} and \emph{non-complete}. 
Nevertheless the Riemannian aspects of special geometry
 have a direct bearing on the physics of $\cn=2$ QFTs. Physical intuition suggests
 that its geodesics enjoy distinctive properties, and one wishes to confirm (or refine)
 these expectations by direct geometric proofs.
 
In the much deeper context of `projective' special K\"ahler geometries which describe
 4d $\cn=2$ SUGRAs arising from  Calabi-Yau compactifications
of Type IIB, it was conjectured in ref.\!\cite{R15} that two distinct geodesics cross at most at one point when pulled back to the smooth simply-connected cover $\mathscr{X}$ of the moduli space $\cm$. The statement
was argued from a deep Quantum Gravity perspective in the context of the swampland program \cite{R15a,R15b,R15c,R15d,R15e,R15f}.
In this note we mainly consider the situation in absence of gravity where the issue is much simpler. 
However the proofs are geometrically elegant so, perhaps, worth reporting in writing.
At the end of the paper we shall briefly mention the kind of results one obtains for the gravitational Weil-Petersson geometry using our methods: while not very powerful, they may still be of interest.

\medskip

Let us motivate the geometric problem from some physical heuristics.
Suppose we have a 4d $\cn=2$ QFT which at low energy is effectively described  by
a special K\"ahler manifold $\cm$ of the kind which ``arises in physics''. 
We may wish to compactify our QFT on a circle $S^1$ to construct
a vacuum configuration of a 3d theory. When the vector-multiplet scalars are non-constant along the circle, this field configuration is a non-trivial map $S^1\to \cm$.
To be a solution of the equations of motion this map should be a geodesics of $\cm$.

One may be more or less conservative about which class of geodesics is admissible in this construction. 
In this regard there are \emph{three} issues to consider. First issue:
 to be a semiclassical vacuum it is not enough that the configuration is a stationary solution of the equations of motion, i.e.\! that the first variation of the action vanishes. This insufficiency is well known
for models with potentials of the form
\be
V(\phi)=\lambda(\phi^2-a^2)^2
\ee  
where the first variation $V^\prime(\phi)$ vanishes for three static configurations, $\phi=0$ and $\phi=\pm a$,
but the first one is not a valid vacuum since it is unstable (it contains tachyons) and should be ruled out.
Two conditions should be satisfied to have a semiclassical vacuum: (1) the first variation of the action should vanish, \emph{and} (2) the second
variation must be positive-definite. This criterion applies also to 3d compactifications:
the conservative position is that
the first variation of the action should vanish, i.e.\! $\phi\colon S^1\to \cm$ must be a geodesic,
but also the second variation must be positive semi-definite, i.e.\! the geodesic must be \emph{stable,} again to avoid tachyons. 
Second issue: one may or not allow the geodesic to pass through a bad\footnote{\ By \emph{bad} singularity we mean a point where the curvature diverge. Milder kind of singularities (such as finite quotient orbifolds) are accepted.} singular point in $\cm$.
The conservative position is, of course, to disregard the singular configurations.
Third issue: the moduli space $\cm$ has the form\footnote{\ $\cm$ is the complement of the divisors in the Coulomb branch with local monodromy of infinite order. It is an orbifold with mild quotient singularities.}
\be
\cm=\mathscr{X}/\Gamma
\ee
where $\mathscr{X}$ is a smooth contractible complex manifold (the ``marked moduli space'' in the language of \cite{R15})
and $\Gamma\subset Sp(2r,\Z)$ is the monodromy group\footnote{\ $r$ is the complex dimension of the special K\"ahler manifold.} which acts on $\mathscr{X}$ by (holomorphic) isometries\footnote{\ The action may have fixed points, but the isotropy group of each point is at most finite.}
and on the lattice $\Lambda_\text{e.m}$ of electro-magnetic
charges by generalized electro-magnetic dualities \cite{R8}. A closed geodesic $\gamma$ in $\cm$
may lift to an open one in $\mathscr{X}$. If this is the case, when we go around the circle $S^1$ 
we get back to the starting point $p_0\in\cm$ with a non-trivial redefinition of what we mean by electric and magnetic charges. This corresponds to inserting
a ``duality defect'' in the circle. The most conservative attitude is to limit ourselves to
constructions of ``3d vacua'' without ``duality defects''. Of course, one may (or even should)
be more general and allow for such non-trivial insertions. If such constructions exist, the physics of the resulting 3d effective theory is quite subtle and not covered by the most naive heuristics we are advocating.

Thus, taking the most conservative stance, we consider non-trivial 3d vacua of the 4d $\cn=2$ QFT
 produced by stable geodesics which are closed in the covering space $\mathscr{X}$.
 The physical heuristics is that we do \emph{not} expect a vacuum with these properties to exist,
 except -- possibly -- when a ``duality defect'' is inserted, and also in that case only under
 very special circumstances which often require a larger supersymmetry $\cn\geq3$.
Physics thus makes 
 a geometric prediction:
in the covering space $\mathscr{X}$ there is no stable closed geodesic, and
even on the actual moduli $\cm$ such geodesics are extremely restricted, and their existence
requires peculiar properties which are just false for the \emph{generic} QFT.
This note is dedicated to proving a stronger property of the geodesics in special K\"ahler geometry   
(more in the spirit of \cite{R15}): 

\begin{pro}\label{mainP} In special K\"ahler geometry
two distinct stable complete geodesics of the covering space $\mathscr{X}$ cross at most once.
In particular there are no closed stable geodesics.
\end{pro}  

We stress that the adjective ``stable'' cannot be omitted from the statement. By definition a non-stable geodesic
crosses multiple times a nearby geodesic, so the qualification ``stable'' could be omitted if and only if \emph{all} complete geodesics are automatically stable. This is not the case: indeed ``most'' of the
complete geodesics are unstable.    

\begin{rem} In the Quantum Gravity context of the special geometry arising from the WP geometry of Calabi-Yau moduli, the authors of \cite{R15}  state a \textbf{Conjecture} which is the above \textbf{Proposition} without the qualifications ``stable'' and ``complete''. In particular this implies:
\begin{conj} 
All Weil-Petersson geodesics on the moduli space of a CY 3-fold are stable.
\end{conj}
While in the QFT case ``most'' (complete) geodesics are \emph{un}stable, the situation in the gravity case looks the opposite way: the geodesics have a ``tendency'' to be stable, and the above \textbf{Sub-conjecture} may well be true. (For non-complete geodesics the situation improves: a geodesic may terminate before becoming unstable \cite{R15}).
\end{rem}

The rest of this note is purely Differential-Geometric. It is organized as follows: in section 2 we review some known and less known elementary facts about the Riemannian aspects of special geometry (in the sense of $\cn=2$ QFT),
and  in section 3 we describe explicitly the stable complete geodesics on $\mathscr{X}$ (and $\cm$).
We briefly comment about other results one gets by the same technique, including in the gravitational case of the Weil-Petersson geometry. In the appendix we discuss the geodesics in the vicinity of a finite-distance singularity arising in complex codimension-1.

\section{Elementary facts}

In this section we collect elementary properties of the special K\"ahler geometries which arise from physics. These results 
are not new, but perhaps not widely known, and some of them have not appeared in print before. We present them from the Riemannian viewpoint.

\medskip

To make the story short and simple, we work directly on the simply-connected cover $\mathscr{X}$,
thus trivializing all global aspects of the geometry. However, to avoid meaningless situations,
we focus on the geometries which ``arise from physics'', that is -- technically -- we assume
the 

\begin{condi}\label{first} Let $r\equiv\dim_\C\mathscr{X}$. There is a discrete group $\Upsilon\subset Sp(2r,\Z)$, freely
acting on $\mathscr{X}$ by holomorphic isometries of its special K\"ahler metric, such that
the smooth K\"ahler manifold
\be
\cn\equiv \mathscr{X}/\Upsilon
\ee
is (biholomorphic to) a non-compact quasi-projective algebraic variety.  $\Upsilon\triangleleft \Gamma$ is a finite-index, torsion-free,
normal subgroup of the monodromy group $\Gamma$ and $\cn\to\cm$ is a smooth finite cover of the moduli space $\cm$. The local monodromies at infinity in $\cn$ are unipotent.
Under these conditions, there is a torsion-free \emph{arithmetic} group $U\subset Sp(2r,\Z)$ with $\Upsilon \subset U$.
\end{condi}

\textbf{Condition} is the geometric counterpart to UV-completeness for the underlying $\cn=2$ QFT.
In the geometric analysis it replaces the assumption of \emph{geodesic completeness} which cannot hold unless the special geometry is flat (see below). 

\subsection{Properties of the Riemann tensor}\label{curvv}

We recall the relevant definitions. The Riemann tensor $R_{i\bar j k\bar l}$ of a K\"ahler manifold is
\emph{Nakano non-negative} iff for all
tensors $u^{ij}$
\be
R_{i\bar j k\bar l}\,u^{ik}\bar u^{\bar j\bar l}\geq0.
\ee
 Nakano non-negativity implies Griffiths non-negativity, non-negative holomorphic bisectional
 curvatures, and  
 non-negative Ricci tensor
(see e.g.\! \cite{R16,R21}). A smooth function $f$ on a complex manifold is said to be
\emph{pluri-subharmonic} iff the $(1,1)$ part of its Hessian is everywhere a non-negative matrix, i.e.
\be
\partial_i\mspace{1mu}\partial_{\bar j}f\geq0.
\ee
When $\mathscr{X}$ satisfies \textbf{Condition}, a $\Upsilon$-invariant\footnote{\ I.e.\! a function which factors through $\cn$.} pluri-subharmonic function $f\colon \mathscr{X}\to \R$
is bounded if and only if it is constant. 

\medskip

A special geometry is characterized by the presence of a holomorphic \emph{symmetric} 3-form
$Y\in \Gamma(\mathscr{X},\odot^3\mspace{1mu} T^*\mspace{-3mu}\mathscr{X})$
\be
Y\equiv Y_{ijk}\, dz^i\otimes dz^j\otimes dz^k\qquad \overline{\partial}Y=0,
\ee
defined by the Pauli couplings in the low-energy effective Lagrangian
\be
Y_{ijk}\, \overline{\lambda}^i_L\sigma^{\mu\nu}\lambda^j_L \,F_{\mu\nu}^{+\,k}+\text{h.c.},
\ee
where the $\lambda_L^i$ are left-handed gauginos and the $F_{\mu\nu}^{+\,k}$
are the (complex) self-dual components of the gauge field-strengths. See footnote \ref{footN} for the interpretation of $Y$ in terms of holomorphic integrable systems.\footnote{\ \label{footN}In the language of holomorphic integrable systems \cite{R7,R13} i.e.\! of holomorphic fibrations $\mathscr{A}\to\mathscr{X}$ with Lagrangian fibers which are polarized Abelian varieties, $Y$ is the Kodaira-Spencer map 
$\phi_\text{KS}\colon T\mspace{-2mu}\mathscr{X}\otimes \Lambda\to \Lambda^\vee$ of the underlying family of Abelian varieties. (Here $\Lambda\to \mathscr{X}$ is the holomorphic vector bundle whose fibers are the $H^{1,0}$ groups of the Abelian fibers). Since the symplectic structure yields an isomorphism $\Lambda\simeq T\mspace{-2mu}\mathscr{X}$, the Kodaira-Spencer map is a 3-form on $T\mspace{-2mu}\mathscr{X}$ which is easily seen to be symmetric.}  
In local special coordinates $a^i$,
with holomorphic pre-potential $\cf(a)$,
\be
Y= -i\,\frac{\partial^3 \cf(a)}{\partial a^i\partial a^j\partial a^k}\, da^i\otimes da^j\otimes da^k.
\ee
A simple local computation yields
\be\label{xrtttt}
R_{i\bar j k\bar l}= G^{\bar n m}\, Y_{ik m} \bar Y_{\bar j\bar l\bar n}.
\ee
where $G^{\bar i j}$ is the inverse of the special K\"ahler metric $G_{i\bar j}$.
\begin{corl}\label{coorr} {\bf(1)} The Riemann tensor of a special K\"ahler geometry is Nakano  non-negative, hence its holomorphic bisectional and Ricci curvatures are non-negative. {\bf(2)} The scalar curvature
is $R=\|Y\|^2$ and hence if $R$ is zero at the point $p$, the tensor $Y$ vanishes at $p$, and therefore the full
Riemann tensor $R_{i\bar j k\bar l}$ is zero at that point. {\bf(3)} The scalar curvature $R$
is a pluri-subharmonic function. 
\end{corl}
Only {\bf(3)} requires a proof. One has\footnote{\ Here 
$\langle Y, Y\rangle\equiv Y_{ijk}G^{i\bar l}G^{j\bar m} G^{k\bar n}\bar Y_{\bar l \bar m \bar n}\equiv\|Y\|^2$
is the Hermitian norm of $Y$; our convention is that the Hermitian forms are anti-linear in their \emph{second} argument.}
\begin{equation}
\begin{split}
\partial_{\bar j}\partial_i R&= \langle \nabla_{\bar j}\nabla_i Y,Y\rangle+
\langle \nabla_i Y,\nabla_j Y\rangle
=\langle [\nabla_{\bar j},\nabla_i] Y,Y\rangle+
\langle \nabla_i Y,\nabla_j Y\rangle\geq 0
\end{split}\label{poiuye}
\end{equation}
since both terms in the \textsc{rhs} are non-negative (the first one because $R_{i\bar j k\bar l}$ is Griffiths non-negative).
Therefore, under \textbf{Condition}, the scalar curvature $R$ is either constant or unbounded above.

\begin{lem} Assume {\bf Condition}. If the scalar curvature $R$ is constant, the special K\"ahler metric is locally flat.
\end{lem}
\begin{proof}
If $R$ is constant, both terms in the \textsc{rhs} of eq.\eqref{poiuye} vanish.
In particular one has
\begin{equation}
0= G^{i\bar j}\langle [\nabla_{\bar j},\nabla_i]Y, Y\rangle= 3\, R^{m\bar n} \,Y_{mkl} \bar{Y}_{\bar n \bar k \bar l}\,G^{k\bar k} G^{l\bar l}= 3\, R^{m\bar n}R_{m\bar n}
\end{equation}
which implies that the manifold is Ricci flat. But a Ricci flat special K\"ahler manifold is flat by \textbf{Corollary \ref{coorr}(2)}.
\end{proof}

\begin{corl}[\bf Dycothomy] In a special K\"ahler geometry $\mathscr{X}$ satisfying {\bf Condition:}
\begin{itemize}
\item either $R_{i\bar j k\bar l}\equiv0$ everywhere, i.e.\! $\mathscr{X}$ is flat,
\item or $\sup_{\mathscr{X}} R=+\infty$.
\end{itemize}
In particular a homogeneous special K\"ahler manifold is locally flat.
\end{corl}

Moreover one can show \cite{R7,R7.5} that the points where $R$ diverges should be at finite distance from regular points of $\mathscr{X}$. Thus a special K\"ahler manifold is
non-complete unless it is (locally) flat. An alternative proof  follows from the results in \S.\,\ref{stableG}.

\subsection{Relation with the period maps}

The gauge couplings define holomorphic period maps
\begin{align}\label{target}
&p\colon\cm \to \Gamma\backslash Sp(2r,\R)/U(r)\\
& \check{p}\colon \cn\to U\backslash Sp(2r,\R)/U(r)\label{target2}
\end{align}
which may be lifted to a holomorphic map to Siegel's upper half-space
\be
\tau\colon \mathscr{X}\to Sp(2r,\R)/U(r)\equiv \mathscr{H}
\ee
given locally in terms of special coordinates $a^i$ as
\be
a^k\mapsto \tau_{ij}\equiv\frac{\partial^2\cf(a)}{\partial a^i\partial a^j}
\ee
where we see $\mathscr{H}\equiv Sp(2r,\R)/U(r)$
 as the space of complex symmetric matrices $\boldsymbol{\tau}=(\tau_{ij})$ with positive-definite imaginary part. 
 $\mathscr{H}$ is a K\"ahler symmetric space
 and a Hadamard manifold, i.e.\!
a complete simply-connected Riemannian manifold with non-positive sectional curvatures.
In particular $\mathscr{H}$ is Einstein with a negative Ricci tensor.
We normalize the symmetric metric on $\mathscr{H}$ so that its global K\"ahler potential is
\be\label{pioue1}
K_\mathscr{H}(\boldsymbol{\tau},\boldsymbol{\bar\tau})=-\log\det\mathrm{Im}\,\boldsymbol{\tau}.
\ee
The Ricci form $\mathsf{Ric}_\mathscr{H}$ is then related to the symmetric K\"ahler form $\omega_\mathscr{H}$ of $\mathscr{H}$
as
\be\label{whillla}
\mathsf{Ric}_\mathscr{H}=-(r+1)\,\omega_\mathscr{H}.
\ee
In local special coordinates the K\"ahler metric takes the form
\be\label{pioue2}
ds^2= \mathrm{Im}(\tau_{ij})\,da^i \,d\bar a^{\bar j}.
\ee 
Since $\tau$, $p$ and $\check{p}$ are holomorphic maps between K\"ahler spaces, they are in particular harmonic maps from a manifold with non-negative Ricci curvature to a manifold with non-positive sectional curvatures. This entails that their energy density is sub-harmonic \cite{R16.5}. 
At the global level the period maps $\tau$, $p$ and $\check{p}$ are further restricted by the structure theorem for variations of Hodge structure (VHS), see \cite{R17,R18}.

The formulae for the curvatures in \S.\,\ref{curvv} may be written more functorially as:

\begin{corl} {\bf(1)} Let $\mathsf{Ric}_{\mathscr{X}}$ be the Ricci form of the special K\"ahler metric on $\mathscr{X}$.
One has
\be\label{poiuq}
\mathsf{Ric}_{\mathscr{X}}=\tau^*\omega_\mathscr{H}.
\ee
{\bf(2)} the scalar curvature $R$ is the energy density $e(\tau)$ of the harmonic map $\tau$.
{\bf(3)} A special geometry is flat if and only if $\tau$ (hence $p$, $\check{p}$) is a constant map. 
\end{corl}
\begin{proof} Compare eqs. \eqref{pioue1},\,\eqref{pioue2}.\end{proof}
Thus, say,
\be
\int_{\cn} \sqrt{g}\,R\;d\mathsf{vol}= \mathbf{E}(\check{p}),
\ee
where $\mathbf{E}(\cdot)$ is the energy\footnote{\ Suitably normalized.} in the sense of harmonic maps \cite{R19}, or, in the physics language, the classical action
of a $\sigma$-model with source and target spaces as in \eqref{target2}. Thus $\tau$ is the map with
minimal energy (action) between the ones equivariant under the action of $\Upsilon\subset Sp(2r,\Z)$ on $\mathscr{X}$
and $\mathscr{H}$. This statement is essentially equivalent to the structure theorem for $\check{p}$ \cite{R20}.

\begin{exe}\label{exam} Suppose we have a simple closed geodesic $\gamma$ in a non-flat special geometry of dimension 1. (When the geometry is \emph{non} simply-connected they typically exist, but should be unstable by \textbf{Proposition \ref{mainP}}). The interior $\Sigma$ of the geodesics is mapped by $\tau$ into a contractible region $\tau(\Sigma)$ of the half-plane $\mathscr{H}$. Together with eq.\eqref{poiuq}, the Gauss-Bonnet theorem  implies that the area of $\tau(\Sigma)$
is $2\pi$, so that, by \eqref{whillla},
\be
\int_{\tau(\Sigma)} \mathsf{Ric}_\mathscr{H} =-4\pi.
\ee
Applying the Gauss-Bonnet theorem to the image region $\tau(\Sigma)$, we get that the integral of the geodesic curvature $\kappa_g$
along the closed curve $\tau(\gamma)$ is
\be
\int_{\tau(\gamma)}\kappa_g\,ds=6\pi,
\ee 
which means that the image curve in $\mathscr{H}$ is rather far from being a geodesic, quite the opposite. 
This implies that, given a singular point $p\in\cn$, there is a radius $R_0$ so that all disks
$B(p,r)$ with $r<R_0$ do not contain any closed geodesic wrapping around $p$, while any larger ball typically contains infinitely many of them. See appendix for more details.  
\end{exe}

\subsection{Symmetries}

To cover the special case $\dim_\C\mathscr{X}=1$, we define $\mathsf{Conf}^0(\mathscr{X})$ to be the connected Lie group of holomorphic motions generated by vectors $V$ satisfying
\be\label{iuyt}
\mathscr{L}_V\omega_{\mspace{-1mu}\mathscr{X}}=c\cdot\omega_{\mspace{-1mu}\mathscr{X}}
\ee
with $c$ a constant. The connected isometry group $\mathsf{Iso}^0(\mathscr{X})\subset\mathsf{Conf}^0(\mathscr{X})$ is the Lie group generated by vectors $V$ satisfying \eqref{iuyt} with $c=0$.
In Riemannian geometry a vector $V^\alpha\partial_\alpha$ such that $g_{\alpha\beta}=\nabla_\alpha V_\beta$ is called \emph{concurrent}: a Riemannian manifold is isometric to a \emph{cone}
iff it has a concurrent vector \cite{yano}; in a K\"ahlerian cone the vector $IV$ ($I$ being the 
complex structure) is a Killing vector.

\medskip

With this definition, a holomorphic conformal symmetry
leaves the Ricci form $\mathsf{Ric}_\mathscr{X}$ invariant.
In view of \eqref{poiuq} this yields:

\begin{corl} {\bf(1)} $K$ a holomorphic Killing vector of the special K\"ahler geometry. $\tau_\ast K$, if non-zero,
is a Killing vector of the symmetric metric on $\mathscr{H}$ i.e.\! $\tau_\ast K\in \mathfrak{sp}(2r,\R)$.
{\bf (2)} More generally, if $\dim_\C\mspace{-2mu}\mathscr{X}\geq2$ and $V$ is vector field generating a holomorphic conformal motion of $\mathscr{X}$, $\tau_\ast V$ 
 is a Killing vector of $\mathscr{H}$.
\end{corl}

\begin{rem}\label{info}
Informally this \textbf{Corollary} says that, by supersymmetry, a continuous symmetry of the vector scalars' kinetic terms  is also a symmetry (possibly trivial) of the vectors' kinetic terms: it is a compatibility condition between SUSY and bosonic symmetries.
\end{rem}

\medskip

Let $\boldsymbol{\tau}$ be a generic point in the image $\tau(\mathscr{X})$. Abusing language we refer to
the complex submanifold $\tau^{-1}(\boldsymbol{\tau})\subset\mathscr{X}$ as the \emph{$\tau$-fiber} $\mathscr{X}_{\mspace{-1mu}\boldsymbol{\tau}}$ over $\boldsymbol{\tau}$. Same for the map $\check{p}$.

\begin{lem} Assume {\bf Condition.} The $\check{p}$-fibers have finitely many connected components.
\end{lem}

\begin{proof}
Under \textbf{Condition} we can find a pair $(\overline{\cn},D)$ with $\overline{\cn}$
a smooth projective variety and $D\subset\overline{\cn}$ a normal crossing divisor, such that $\cn=\overline{\cn}\setminus D$.
By Borel's extension theorem (\textbf{Theorem 3.10} of \cite{Bor}), $\check{p}$ extends to a
morphism of projective varieties from $\overline{\cn}$ to the  Baily-Borel compactification \cite{BaiB}
of the arithmetic quotient $U\backslash Sp(2r,\R)/U(r)$. The statement then follows from Stein factorization (cf.\! \textbf{Corollary 11.5} of \cite{Hart}).
\end{proof}

\begin{pro}\label{yttttx} Suppose the special K\"ahler manifold $\mathscr{X}$ satisfies {\bf Condition}.
Let $\mathsf{Iso}^0(\mathscr{X})$  be the connected component of its holomorphic isometry group.
Each $\mathsf{Iso}^0(\mathscr{X})$ orbit  is fully contained in a $\tau$-fiber $\mathscr{X}_{\boldsymbol{\tau}}=\tau^{-1}(\boldsymbol{\tau})$. In other words: the period matrix $\tau_{ij}$ is constant along each orbit of the isometry group.  When $\dim_\C\mspace{-2mu}\mathscr{X}\geq2$ the same statement holds for the orbits of 
the connected component $\mathsf{Con}^0(\mathscr{X})$ of the (holomorphic) conformal symmetry group. 
\end{pro}

\begin{rem} \textbf{Condition} is necessary. Otherwise it is easy to construct explicit counterexamples: e.g.\! the special geometry \eqref{ghdf} in the appendix.
\end{rem}

\begin{proof} Let $\mathfrak{k}$ be the Lie algebra of conformal motions. $\tau_\ast(\mathfrak{k})\subset\mathfrak{sp}(2r,\R)$ is a subalgebra. Since $\Upsilon\subset Sp(2r,\Z)$ acts by isometries,
it normalizes $\mathfrak{k}$ hence $\tau_\ast(\mathfrak{k})$ since $\tau$ is $\Upsilon$-equivariant.
$\tau_\ast(\mathfrak{k})$ is then normalized by the Zariski closure $G(\R)$ of $\Upsilon$ in $Sp(2r,\R)$.
$G(\R)$ is semisimple and we may assume it simple with no loss.
$\tau_\ast(\mathfrak{k})\subset\mathfrak{g}(\R)$ (by the structure theorem of VHS), and then
either  $\tau_\ast(\mathfrak{k})=0$ or $\tau_\ast(\mathfrak{k})=\mathfrak{g}(\R)$. In the first case
$\boldsymbol{\tau}$ is constant along the $\mathsf{Iso}^0(\mathscr{X})$-orbits and we are done.
In the second case, given a point $q\in\mathscr{X}$ there is an isometry which maps it in some point of a chosen reference fiber. Therefore the Ricci curvature is bounded everywhere, so the full space should be  flat: this is a contradiction since $\tau$ is not the constant map. Only the first case may happen.
\end{proof}

\begin{pro}\label{uuuaae} Assume {\bf Condition}. The orbits of $\mathsf{Conf}^0(\mathscr{X})$
are complex analytic submanifolds of $\mathscr{X}$. Moreover the complex Lie group $\mathsf{Conf}^0(\mathscr{X})$ is Abelian.\end{pro}

We defer the proof to \S... where more details will be provided.

\subsection{Geometry of the $\tau$-fibers}

We need to understand the geometry of the generic fibers of the period maps $\tau$ and $\check{p}$.
The 
connected components of the fibers of $\tau$ are finite covers of the fiber of $\check{p}$, and hence are smooth  quasi-projective algebraic varieties. For a generic special K\"ahler geometry the connected components of the fibers are trivial i.e.\! points.
However we know two important situations where the fiber has
positive dimension.
\begin{exe}\label{X1} The geometry is \emph{isotrivial.} This is automatically true when SUSY is enhanced to $\cn\geq3$.
More generally, it holds \emph{if and only if} at all points in the Coulomb branch the locally light states are effectively described by a SCFT (as contrasted to an IR-free effective theory), see \cite{R7,R7.5}. 
In this case $\tau$ and $\check{p}$ are constant maps, the full special K\"ahler manifold $\mathscr{X}$ is the fiber, and the fiber is flat. The isometry group $\mathsf{Iso}^0(\mathscr{X})$ (and \emph{a fortiori} $\mathsf{Conf}^0(\mathscr{X})$) acts transitively on the fiber $\equiv\mathscr{X}$.
\end{exe}
\begin{exe}\label{X2}
A generic (non-isotrivial) $\C^\times$-isoinvariant geometry (which describes a \emph{superconformal}
$\cn=2$ QFT). The connected components of the fibers have complex dimension 1 and coincide
with the orbits in $\mathscr{X}$ of the superconformal
symmetry $\C^\times$, that is, the symmetry group $\mathsf{Conf}^0(\mathscr{X})$ acts transitively on the fibers which are also flat for the metric induced on the orbits by the ambient special K\"ahler metric. 
\end{exe}  

Recall that a \emph{flat} in a Riemannian manifold is a submanifold which is totally geodesic while the metric induced by the ambient space is flat. The existence of flats of real dimension $\geq2$ is quite constraining for a Riemannian manifold.
We observe that the orbits of $\mathsf{Conf}^0(\mathscr{X})$ in \textbf{Examples \ref{X1},\,\ref{X2}} are \emph{holomorphic flats} (i.e.\! complex submanifold which are also Riemannian flats).

\begin{lem}\label{uuuu123} Let $i_{\boldsymbol{\tau}}\colon \mathscr{X}_{\boldsymbol{\tau}}\hookrightarrow \mathscr{X}$ be a general $\tau$-fiber.
\begin{itemize}
\item[\bf(1)] The restriction $i_{\boldsymbol{\tau}}^*\mspace{1mu}T\mspace{-2mu}\mathscr{X}\to\mathscr{X}_{\boldsymbol{\tau}}$  of the holomorphic tangent bundle  to $\mathscr{X}_{\boldsymbol{\tau}}$ is \emph{flat} for the Chern connection $i^\ast_{\tau(p)}\!\nabla^\mathscr{X}$ (the restriction of the Levi-Civita connection $\nabla^\mathscr{X}$\! on $\mathscr{X}$\!);
\item[\bf(2)] Equip the submanifold $\mathscr{X}_{\boldsymbol{\tau}}\subset\mathscr{X}$ with the induced K\"ahler metric.
{\bf(2a)} The Riemann tensor of $\mathscr{X}_{\boldsymbol{\tau}}$ is Nakano non-positive, hence its Ricci form and biholomorphic sectional curvatures are non-positive. {\bf(2b)} The induced K\"ahler metric on $\mathscr{X}_{\boldsymbol{\tau}}$ is flat if and only if the fiber $\mathscr{X}_{\boldsymbol{\tau}}$ is totally geodesic (hence $\mathscr{X}_{\boldsymbol{\tau}}$ is a \emph{holomorphic flat}).
\end{itemize}
\end{lem}

\begin{proof} {\bf(1)} the curvature of the pulled back tangent bundle vanishes.
Indeed if $V\in T\mathscr{X}_{\mspace{-1mu}\boldsymbol{\tau}}$
\be\label{juytq22}
V^i R_{i\bar j k\bar l} = G^{m\bar n} V^i Y_{imk}\bar Y_{\bar j \bar n\bar l}=0
\ee
since (working in special local coordinates) $V^i Y_{ijk}=-iV^i \partial_j\tau_{jk}=0$.
Hence \emph{a fortiori}
\be
R(V,\bar W)_{i \bar j} =V^k \bar W^{\bar l} G^{m\bar n} Y_{imk}\bar Y_{\bar l \bar n\bar k}=0\qquad\text{for }\ V, W\in T\mathscr{X}_{\mspace{-1mu}\boldsymbol{\tau}},
\ee
{\bf (2a)} Since the curvatures of holomorphic bundles are Nakano-monotonic in 
sub-bundles \cite{R16,R21}, we obtain that the sub-bundle is Nakano non-positive. {\bf(2b)} the difference between the curvature of a holomorphic bundle restricted to a sub-bundle and the curvature of the sub-bundle is the Hermitian square of the second fundamental form, so when the difference vanishes also the second fundamental form is zero and the submanifold is totally geodesic.  
\end{proof}

Part {\bf(2a)} says that the curvature is non-positive in the holomorphic sense.
Unfortunately, this result comes short of saying that the underlying Riemannian sectional curvatures are 
non-positive. However this holds when the fiber has complex dimension 1.

\begin{corl}\label{diim1} If $\dim_\C\mspace{-1.5mu}\mathscr{X}_{\mspace{-1mu}\boldsymbol{\tau}}\leq1$, the fiber $\mathscr{X}_{\mspace{-1mu}\boldsymbol{\tau}}$ (with induced metric) is a Hadamard manifold.
\end{corl}

 The situation in the two \textbf{Examples} above is far better: the fibers are \emph{flats}
 not just non-positively curved.
Physical intuition says that this must be the general situation:
the idea is that \textbf{Remark \ref{info}} can be reversed, and an invariance of the vector's kinetic terms should, by supersymmetry, be also an invariance of the vector scalars' kinetic terms.Then the $\tau$-fibers are the orbits of the symmetries, and also flats in the geometry.
Now we prove some geometric results which corroborate this physical intuition.

\subsection{Symmetry and structure}

We start by proving \textbf{Proposition \ref{uuuaae}}. In Riemannian geometry a vector
$V$ satisfying the equation
\be
\mathscr{L}_V g_{\alpha\beta}=\lambda\, g_{\alpha\beta}\quad \lambda\ \text{constant}
\ee
satisfies the identity
\be
\nabla_\alpha\nabla_\beta V_\gamma=- R_{\beta\gamma\alpha\delta}V^\delta.
\ee
In our set-up this formula applies to the vectors $V$ which generate $\mathsf{Con}^0(\mathscr{X})$. 
By \textbf{Proposition \ref{yttttx}} they are tangent to the $\tau$-fiber, and then the \textsc{rhs}
vanishes by eq.\eqref{juytq22}. Therefore the tensor $\nabla_\alpha V_\beta$ is parallel.
Consider the Riemannian holonomy algebra
$\mathfrak{hol}(\mathscr{X})$. Since Ricci-flat is flat for a special geometry,
\be
\mathfrak{hol}(\mathscr{X})\simeq \bigoplus_{k=1}^m \mathfrak{u}(n_k)\qquad \sum_k n_k = r-f  
\ee
and the special K\"ahler metric is locally\footnote{\ We cannot apply De Rham theorem and conclude that $\mathscr{X}$ is globally a product because it is non-complete in general.} isometric to 
a flat $\C^f$ times a product of $m$ \emph{irreducible} K\"ahler manifolds of dimensions $n_1,\dots,n_m$. The holomorphic isometries of the factor $\C^f$ is $U(f)\ltimes\C^f$
whose orbit $\C^f$ is a complex analytic space. We can analyze the symmetries factor space by factor space. On an irreducible K\"ahler manifold we have only two parallel tensor, the metric $g_{\alpha\beta}$
and the K\"ahler form $\omega_{\alpha\beta}=g_{\alpha\gamma}\mspace{2mu}{I^\gamma}_\beta$.
Therefore in each non-trivial local factor
\be\label{y6xxx}
\nabla_\alpha V_\beta= a\, g_{\alpha\beta}+b\, \omega_{\alpha\beta}\quad a,b\ \text{constants}.
\ee
Then
\be
\nabla_\alpha(IV)_\beta = a\,\omega_{\alpha\beta}-b\, g_{\alpha\beta}\quad\Rightarrow\quad \nabla_\alpha (IV)_\beta+\nabla_\beta(IV)_\alpha= -2b\, g_{\alpha\beta}
\ee
i.e.\! if $V\in\mathfrak{con}(\mathscr{X})$ then $IV\in \mathfrak{con}(\mathscr{X})$
which proves \textbf{Proposition \ref{uuuaae}}.

In each non-flat locally irreducible factor there is no non-zero parallel vector. Hence if $V\neq0$,
in eq.\eqref{y6xxx}
$a$ and $b$  cannot be both zero. Then $IV$ is a non-zero vector $(a^\prime,b^\prime)\equiv (-b,a)$
is linear independent of $(a,b)$. Hence if $W$ is a vector generating a symmetry in the
irreducible factor we can find constants $c_1$ and $c_2$ such that
\be
W-c_1\, V-c_2IV
\ee
is parallel, hence zero. Therefore in each non-flat factor either there is no vector generating a symmetry or there are precisely 2: a concurrent one $\nabla_\alpha V_\beta= g_{\alpha\beta}$
and a Killing one $\nabla_\alpha(IV)_\beta=\omega_{\alpha\beta}$. We conclude
\begin{pro} Assume {\bf Condition}. A special K\"ahler manifold $\mathscr{X}$ is \emph{locally} isometric to a product of
\begin{itemize}
\item a flat space
\item irreducible cones over Sasaki manifolds whose continuous symmetry is \emph{exactly} $\C^\times$
\item irreducible K\"ahler manifolds with no continuous symmetry. 
\end{itemize}
\end{pro}

\section{Stable geodesics}\label{stableG}

We say that a geodesic $\gamma$ is \emph{stable between points $p$ and $q$} if the geodesic arc connecting them is locally length-minimizing between all nearby curves joining $p$ and $q$,
that is, if the Jacobi field operator $\boldsymbol{J}$  is positive-definite on the space
of normal vector fields along $\gamma$ which vanish at both $p$ and $q$. 
We say that $\gamma$ is \emph{stable} if it is stable between all pairs of its points.

\subsection{Complete stable geodesics}

A geodesic on a Riemannian manifold $X$ is \emph{complete} iff it can be extended to a a geodesic $\gamma\colon(-\infty,+\infty)\to X$
 where, as always, we parametrize geodesics with their arc-length. In particular closed geodesics are complete.
 
 \medskip
 
 We consider geodesics $\gamma$ on the simply-connected special K\"ahler manifold
  $\mathscr{X}$ which are both complete and stable. 
   
 \begin{pro}\label{xxxy} In a special K\"ahler manifold $\mathscr{X}$ -- not necessarily satisfying {\bf Condition} -- the stable complete geodesics are fully contained in a $\tau$-fiber $\mathscr{X}_{\mspace{-1mu}\boldsymbol{\tau}}$. 
 They are also geodesics for the induced K\"ahler metric on $\mathscr{X}_{\mspace{-1mu}\boldsymbol{\tau}}$.
 \end{pro}
 
We defer the proof to \S.\,\ref{s:mayer} after some preparation.
Here we discuss a few first consequences.
 Together with \textbf{Corollary \ref{diim1}}, the \textbf{Proposition} gives:
 
 \begin{corl} If the
 complex dimension of the $\tau$-fibers $\mathscr{X}_{\mspace{-1mu}\boldsymbol{\tau}}$ is $\leq1$, two distinct complete stable geodesics in $\mathscr{X}$ cross at most once. In particular: in this case there is no stable closed geodesic.
 \end{corl}
 This already settles the issue for both generic special K\"ahler geometries and 
 generic $\C^\times$-isoinvariant special geometries. On the opposite extremum, \emph{isotrivial} geometries,
 $\mathscr{X}$ is just the Euclidean space $\R^{2r}$ and the statement is trivially true.
 This observation suffices to cover all geometries with $\dim_\C\mspace{-2mu}\mathscr{X}\leq2$.
 
However the physical intuition mentioned in the \textbf{Introduction} suggests that the statement holds more in general. Moreover the last sentence of \textbf{Proposition \ref{xxxy}} holds automatically when the $\tau$-fibers are flats. We next show that the physical intuition is correct, at least about the 
number of crossing of complete stable geodesics.

  \begin{pro}
  In any special K\"ahler geometry two distinct stable complete geodesics cross at most once. In particular there are no stable closed geodesics.
  \end{pro}
  
 \begin{proof} Suppose we have two distinct geodesics $\gamma_1$, $\gamma_2$ which cross at one point $p\in\mathscr{X}$. Then by \textbf{Proposition \ref{xxxy}} $\tau(\gamma_1)=\tau(\gamma_2)=\tau(p)$ and the two geodesics are both fully contained in the same fiber $\mathscr{X}_{\tau(p)}$ and their tangent vectors
 \be
 \dot\gamma_a(s_a)\in T\mathscr{X}_{\tau(p)}\big|_{\gamma_a(s_a)}\subset i_{\tau(p)}^*\mspace{1.5mu}T\mspace{-2mu}\mathscr{X}\big|_{\gamma_a(s_a)}\quad \forall\; s_a,\quad a=1,2.
 \ee
 The fiber $\mathscr{X}_{\tau(p)}$ is simply-connected and
 the holomorphic bundle $i_{\tau(p)}^*\mspace{1.5mu}T\mspace{-2mu}\mathscr{X}$
 is flat for the Levi-Civita connection $i^\ast_{\tau(p)}\!\nabla^\mathscr{X}$ of the total manifold $\mathscr{X}$ (cf.\! \textbf{Lemma \ref{uuuu123}(1)}), so that we have an isomorphism
 \be\label{oppiu}
 i_{\tau(p)}^*\mspace{1.5mu}T\mspace{-2mu}\mathscr{X}\simeq \mathscr{X}_{\tau(p)}\times \C^{r}
 \ee 
 given by $i^\ast_{\tau(p)}\!\nabla^\mathscr{X}\!$-parallel sections.
 The tangent vectors to the geodesics, $\dot\gamma_1(s_1),\,\dot\gamma_2(s_2)$, are $i^\ast_{\tau(p)}\!\nabla^\mathscr{X}\!$-parallel
 for all times $s_1,\,s_2$, so they project to constant vectors $v_1,\,v_2$ in the factor $\C^r$ of \eqref{oppiu}. At the starting point the vectors
 $\dot\gamma_1(0)\equiv v_1$ and $\dot\gamma_2(0)\equiv v_2$ diverge because the two geodesics are leaving the point $p$ in different directions. Since the vectors $v_1$, $v_2$ are constant, they remain diverging at all times $s_1,\,s_2$, and 
 the two distinct geodesics cannot converge back to the same point. 
  \end{proof}
  
  \begin{rem} The tangent vectors $\dot\gamma_a(s_a)$ are parallel also for the sub-bundle connection
  on $T\mspace{-2mu}\mathscr{X}_{\tau(p)}$. Thus the existence of a complete stable geodesic in $\mathscr{X}$
  yields conditions on the second fundamental form of the fiber containing it. In particular a fiber of complex dimension 1 which contains a geodesic is a flat of $\mathscr{X}$. (Cf.\! the physical intuition above).
  \end{rem} 
  
  \begin{rem} Replacing the cover $\mathscr{X}$ with the enhanced moduli $\cn$,
  the corresponding bundle $i_{\check{p}(p)}^*\mspace{1.5mu}T\mspace{-2mu}\mathscr{\cn}$ 
  over the $\check{p}$-fiber will still be flat but not globally trivial because of monodromy.  
  \end{rem}

\subsection{Elementary bounds on Ricci integrals}

Let $\xi\colon[a,b]\to\mathscr{X}$ be any smooth curve in the special K\"ahler manifold $\mathscr{X}$ parametrized
by its arc-length $s$, while $f\colon [a,b]\to \R_{\geq0}$ is any smooth \emph{positive} function.
We are interested in a lower bound to the integral
\be\label{eerrqwe}
I[\xi,f]\equiv\int_a^b ds\, f(s)\,\mathrm{Ric}(\dot\xi,\dot\xi).
\ee 
where $\mathrm{Ric}(\dot\xi,\dot\xi)\equiv R_{i\bar j}\,\dot\xi^i(s)\,\dot\xi^{\bar j}(s)$ and $R_{i\bar j}$ is the Ricci curvature of the special K\"ahler metric. Suppose first that $f\equiv1$.
We write $\eta(s)$ for the image curve $\tau(\xi(s))$ in the Siegel upper half-space $\mathscr{H}\equiv Sp(2r,\R)/U(r)$.
From eq.\eqref{poiuq} we get
\be
I[\xi,1]=\int_a^b ds\,\boldsymbol{G}_{\alpha\bar\beta}\, \dot\eta^\alpha \dot\eta^{\bar\beta}=\mathbf{E}[\eta]
\ee
where $\boldsymbol{G}_{\alpha\bar\beta}$ is the symmetric (K\"ahler) metric in the Siegel upper half-space and $\mathbf{E}[\eta]$ is the energy (action for physicists) of a particle, moving in the Siegel space $\mathscr{H}$,
along the curve $\eta(s)$ from the point $\eta(a)$ at time $s=a$ to the point $\eta(b)$ at $s=b$. 
In the space of all paths $\rho\colon[a,b]\to\mathscr{H}$ satisfying the boundary conditions $\rho(a)=\eta(a)$,
$\rho(b)=\eta(b)$, the minimum of the action
$\mathbf{E}[\rho]$ is attained when $\rho(s)$ is the solution to the Euler-Lagrange equations: the particle moves along a geodesic in $\mathscr{H}$ which connects $\eta(a)$ to $\eta(b)$ at constant speed. The target space $\mathscr{H}$ is
Hadamard, so there is precisely \emph{one} geodesic connecting any two points which is stable and length-minimizing: the length of the arc of this geodesic is equal to the distance $\boldsymbol{d}(\eta(a),\eta(b))$
where $\boldsymbol{d}(\cdot,\cdot)$ is the distance function between two points in the symmetric Hadamard manifold $\mathscr{H}$.
The action computed along this solution is the \emph{absolute minimum} of the functional $\mathbf{E}[\rho]$
in the space of curves $\rho(s)$ with boundary conditions $\rho(a)=\tau(\xi(a))$ and $\rho(b)=\tau(\xi(b))$.
The constant velocity is just $\boldsymbol{d}(\eta(a),\eta(b))/(b-a)$. Therefore we get that \emph{for all
paths $\xi\colon[a,b]\to\mathscr{X}$}   
\be
I[\xi,1]\geq \frac{\boldsymbol{d}(\tau(\xi(a),\tau(\xi(b))^2}{b-a}.
\ee
Now we consider the more general integral \eqref{eerrqwe} with $f(s)$ positive.
We reparametrize ``time'' in our auxiliary mechanical problem by replacing
\be\label{piouq}
ds\leadsto dt\equiv \frac{ds}{f(s)}\quad\text{i.e.}\quad \frac{ds}{dt}=f(s).
\ee
Note that the map  given by the solution of the ODE \eqref{piouq}
\be
[a,b]\to [t(a),t(b)],\qquad s\mapsto t(s)\equiv t(a)+\int_{a}^s \frac{ds}{f(s)},
\ee
 is globally
invertible because the function is monotonically increasing since $f(s)>0$. 
Now
\be\label{chichi}
I[\xi,f]=\int_{t(a)}^{t(b)}\,dt\,\boldsymbol{G}_{\alpha\bar\beta}\, \dot\chi^\alpha \dot\chi^{\bar\beta}=
\mathbf{E}[\chi]
\ee
where $\chi\colon[t(a),t(b)]\to \mathscr{H}$ is the reparametrized image curve
\be
\chi(t)=\eta(s(t))\equiv \tau(\xi(s(t)))\qquad \dot\chi(t)\equiv \frac{d\chi}{dt}=f(s)\,\frac{d\eta}{ds}.
\ee
The absolute minimum of the energy functional \eqref{chichi} is attained by the same trajectory as before,
but now covered at a different speed which is constant in time $t$ instead of time $s$.
The total time interval is now 
\be
\Delta t= \int_{t(a)}^{t(b)} dt= \int_a^b \frac{ds}{f(s)},
\ee
 thus, for all path $\xi\colon[a,b]\to\mathscr{X}$ and all positive functions $f(s)$, one has
\be\label{hytqwert}
\int_a^b \frac{ds}{f(s)}\, \int_a^b ds\,f(s)\,\mathrm{Ric}(\dot\xi,\dot\xi) \geq  
\boldsymbol{d}(\tau(\xi(a),\tau(\xi(b))^2.
\ee


\subsection{Proof of {\bf Proposition \ref{xxxy}}: Myers' argument}\label{s:mayer}

 Let $\gamma\colon (-\infty,+\infty)\to \mathscr{X}$ be a complete stable geodesic parametrized by arc-length.
We consider two arbitrary points $p_1$ and $p_2$ on $\gamma$ and call $2\ell$ the length of the geodesics arc between them. We choose the parametrization so that
 $p_1=\gamma(-\ell)$ and $p_2=\gamma(\ell)$. We fix a real number $L>\ell$.
Since $\gamma$ is complete, there exist points
$p_0=\gamma(-L)$ and $p_3=\gamma(L)$. The length of the longer arc from $p_0$ to $p_3$ is $2L$.

Since the geodesic $\gamma$ is stable, in the larger geodesic arc from $p_0$ to $p_3$ there are
no points conjugate to $p_0$, so the Jacobi operator $\boldsymbol{J}$ is non-negative acting on the space
of normal vector fields along $\gamma$ which vanish at $p_0$ and $p_3$.

We now proceed by mimicking the proof of the Myers theorem (see e.g.\! \cite{R19} \textbf{Corollary 6.3.1}). We consider a $(2r-1)$-tuple of vector fields 
along $\gamma$, 
\be
X_i(s)\quad \text{($i=1,\dots,2r-1$)},
\ee
 which form an orthonormal basis in the space of vectors normal to $\dot\gamma$, i.e.\!
\be
\langle \dot\gamma(s),X_i(s)\rangle=0\qquad \langle X_i(s),X_j(s)\rangle=\delta_{ij}\quad\ \forall\;s,
\ee
 and are invariant by parallel transport along
the geodesic, i.e.\! such that $\nabla_{\mspace{-1mu}\dot\gamma}\mspace{2mu}X_i=0$. We focus on the normal vector fields along the longer geodesic arc
\be
Y_i(s)\overset{\rm def}{=}\cos\!\left(\frac{\pi}{2L}s\right) X_i(s),\qquad i=1,\dots,2r-1,
\ee
which vanish at $p_0\equiv\gamma(-L)$ and $p_1\equiv\gamma(L)$. The stability condition gives (no sum over repeated indices!)
\begin{equation}\label{rrerr}
\begin{split}
0\leq \langle Y_i, \boldsymbol{J}Y_i\rangle&\equiv
\int\limits_{-L}^{L} ds\Big(-\langle \ddot Y_i,Y_i\rangle -\langle R(Y_i,\dot\gamma)\dot\gamma,Y_i\rangle\Big)=\\
&=\int\limits_{-L}^{L} ds\; \cos^2\!\left(\frac{\pi}{2L}s\right)\!
\Big(\frac{\pi^2}{4L^2} -\langle R(X_i,\dot\gamma)\dot\gamma,X_i\rangle\Big)=\\
&= \frac{\pi^2}{4L}- \int\limits_{-L}^{L} ds\; \cos^2\!\left(\frac{\pi}{2L}s\right)\langle R(X_i,\dot\gamma)\dot\gamma,X_i\rangle
\end{split}
\end{equation} 
Now we sum \eqref{rrerr} over the repeated index. Since $\dot\gamma, X_1,\dots, X_{2r-1}$ form an orthonormal frame, we have
\begin{equation}
\begin{split}
0\leq \frac{\pi^2}{4L}(2r-1)&-\int\limits_{-L}^{L} ds\; \cos^2\!\left(\frac{\pi}{2L}s\right)\mathrm{Ric}(\dot\gamma,\dot\gamma) \leq
\\
&\leq 
\frac{\pi^2}{4L}(2r-1) - \int\limits_{-\ell}^{\ell} ds\; \cos^2\!\left(\frac{\pi}{2L}s\right)\mathrm{Ric}(\dot\gamma,\dot\gamma)
\end{split}
\end{equation}
where we used that the Ricci form is non-negative and $\ell<L$.
One has
\be
\int\limits_{-\ell}^{+\ell} \frac{ds}{\cos^2\!\left(\frac{\pi}{2L}s\right)}=\frac{4L}{\pi}\tan\!\left(\frac{\pi \ell}{2L}\right),
\ee
so, using the estimate \eqref{hytqwert} with $\xi\equiv \gamma$, $f(s)=\cos^2(\pi s/2L)$,
$a=-\ell$ and $b=\ell$, we get
\be
\boldsymbol{d}(\tau(p_1),\tau(p_2))^2\leq \pi(2r-1)\tan\!\left(\frac{\pi\ell}{2L}\right)
\ee
for all pairs of points $p_1$, $p_2$ on the stable complete geodesic $\gamma\subset\mathscr{X}$ and all real $2L$ larger than the length of the geodesic arc between $p_1$ and $p_2$ in $\mathscr{X}$. The \textsc{lhs} is independent of $L$. Sending $L\to\infty$
the \textsc{rhs} goes to zero, so that $\tau(p_1)=\tau(p_2)$. Since $p_1$, $p_2$ are two arbitrary points in $\gamma$, the full geodesic $\gamma$ is contained in the same $\tau$-fiber $\mathscr{X}_{\mspace{-1mu}\boldsymbol{\tau}}$.
The last statement of the \textbf{Proposition} follows from the fact that $\gamma$ is locally length minimizing with respect to variations in the total space $\mathscr{X}$\!, and hence \emph{a fortiori} is length minimizing with respect to local
deformations restricted to the submanifold $\mathscr{X}_{\mspace{-1mu}\boldsymbol{\tau}}$.

\subsection{Other applications of Myers' argument}

One may apply Myers' argument to other situations and issues, but the results are less dramatic than for stable complete geodesics in the QFT special geometry.

\subsubsection{Improving the bound}

The bound \eqref{hytqwert} is rather poor. It is saturated when
 the curve
$\eta\equiv \tau\circ\xi\subset\mathscr{H}$ is a geodesic
and the rescaled velocity $v(s)\equiv f(s)|\dot\eta(s)|$ is constant.
The curve $\eta(s)\subset\mathscr{X}$ of actual interest,  which is the image of a geodesic in $\mathscr{X}$\!,
has a very different behavior. We saw in \textbf{Example \ref{exam}} that the image curve $\eta\subset\mathscr{H}$ is very far from being a geodesic. To get a less loose lower bound, in eq.\eqref{hytqwert} we replace the distance between the two endpoints,
$\boldsymbol{d}(p_1,p_2)$
by the length $\boldsymbol{L}(\eta;p_1,p_2)$ of the arc of the curve $\eta(s)$ which actually connects them.
 This length is significantly larger than the distance since
the curve is far from being geodesic. 
The rescaled velocity $v(s)$ is not constant, but it is not clear how to use this information to get a significantly improved bound.

\subsubsection{Semi-complete stable geodesics}
We consider a \emph{stable} geodesic $\gamma\colon[0,+\infty)\to \mathscr{X}$ which starts
at a finite-distance singularity for $s=0$ and then goes on forever. We choose four points $p_0=\gamma(\epsilon)$, $p_1=\gamma(\ell_1)$,
$p_2=\gamma(\ell_2)$ and $p_3=\gamma(L+\epsilon)$ (with $0<\epsilon<\ell_1<\ell_2<L+\epsilon$).
We proceed as before imposing that the Jacobi operator $\boldsymbol{J}$ is positive on the normal vector fields along $\gamma$ which vanish at $p_0$ and $p_3$.
We get the inequality
\be
\boldsymbol{L}(\eta;\tau(p_1),\tau(p_2))^2\leq  \frac{\pi(r-\tfrac{1}{2})\sin[\pi(\ell_2-\ell_1)/L]}{\sin[\pi (\ell_1-\epsilon)/L]\,\sin[\pi (\ell_2-\epsilon)/L]}
\ee
valid for all $0<\epsilon <\ell_1<\ell_2<L+\epsilon$. We may take the limit $\epsilon\to0$, set $\ell_2=\kappa L$
with $\kappa<1$ and take the limit $L\to\infty$ at fixed $\kappa$ and $\ell_1\equiv\ell$: we get $(\eta(s)\equiv\tau(\gamma(s))$)
\be
\boldsymbol{L}(\eta;\eta(\ell),\eta(L))\lessapprox \sqrt{(r-\tfrac{1}{2})\frac{L}{\ell}}\qquad \text{for}\ \ L\ggg\ell
\ee
a bound reminiscent of Brownian motion. Analogously, the distance between the image of two points a fixed distance $\delta$ along the geodesic goes to zero as we shift the pair to infinity:
setting $\ell_1=\ell$, $\ell_2=\ell+\delta$ and $L=\alpha \ell$ and taking
$\ell$ large at fixed $\alpha>1$ we get
\be\label{uytqwe}
\boldsymbol{d}(\eta(\ell),\eta(\ell+\delta))\leq\boldsymbol{L}(\eta;\eta(\ell),\eta(\ell+\delta))\lessapprox \sqrt{r-\tfrac{1}{2}}\left(\frac{\pi}{\alpha\,\sin^2(\pi/\alpha)}\right)^{\!\!1/2}\,\sqrt{\frac{\delta}{\ell}}\,,\qquad \ell\ggg\delta.
\ee
The best bound is obtained for $\alpha\approx 2.6953$ which yields $\approx 1.1747$
for the second (constant) factor in the \textsc{rhs}.

\subsubsection{Stable geodesics of finite length} Suppose the stable geodesic $\gamma$ has both ends at 
singular points, so that its length $L$ is finite. We write $\eta\colon(0,1)\to\mathscr{H}$
for the projected geodesic $\tau\circ\gamma\subset\mathscr{H}$ parametrized by arc-length re-normalized so that the total length of $\gamma$ is $1$. For all $0< a <b <1$ we have
\be
\boldsymbol{L}(\eta;\eta(a),\eta(b))^2\leq \pi(r-\tfrac{1}{2}) \frac{\sin(\pi(b-a))}{\sin(a\pi)\,\sin(b\pi)} 
\ee
independently of the details of the particular geodesic (e.g.\! independently of its length $L$). 
Roughly speaking: the image $\eta(s)$ of a stable finite geodesic is ``essentially universal''. 

\subsection{Weil-Petersson geometries of 3-CY moduli}

We briefly discuss the implications of the Myers' argument we used in the QFT set-up
for the SUGRA case, i.e.\! for the Weil-Petersson (WP) geometry of the moduli of Calabi-Yau 3-folds.
The result will be less dramatic than in the ``rigid'' case but, perhaps, still useful.

\medskip

The Myers' argument applies to \emph{stable} geodesics. If the \textbf{Sub-conjecture} holds,
all WP geodesics on the Calabi-Yau moduli are stable, and we can use the argument to study
 geodesics in full generality. However in this case the stability condition will be automatically satisfied, and hence cannot be too restrictive, and therefore will say little about the geodesics. We must expect the implications of the Myers argument to be ``weak''. The optimistic reader may wish to interpret 
 the weakness of Myers' bound in the Weil-Petersson geometry as ``evidence'' for the
\textbf{Sub-conjecture} and hence for its parent \textbf{Conjecture} \cite{R15}. 

\medskip

The main ingredient of the Myers argument for the QFT geometries
 was a holomorphic map $\tau$
from the simply-connected special K\"ahler manifold $\mathscr{X}$ to a K\"ahlerian Hadamard space $\mathscr{H}$
with the remarkable property that the Ricci form on $\mathscr{X}$ can be read from the pull-back via $\tau$ of the 
K\"ahler form on $\mathscr{H}$. In the Weil-Petersson geometry of Calabi-Yau moduli there is a map
with exactly these properties: it is the identity map
\be
\mathsf{id}\colon \mathscr{X}_\text{WP}\to \mathscr{X}_\text{H}
\ee 
seen as a homomorphic map from the covering moduli space $\mathscr{X}_\text{WP}$,
equipped with its Weil-Petersson metric -- which is a SUGRA special K\"ahler manifold --
 to the \emph{same} complex manifold
equipped with the Hodge metric\footnote{\ For the relation of the Hodge metric with the genus-1 2-point functions in topological strings, see \cite{R27b}.} \cite{R18,R23,R24,R25,R26,R27}\!\!\cite{R11} which now is a Hadamard K\"ahler space.
The Ricci form on $\mathscr{X}_\text{WP}$ is related to the WP and Hodge K\"ahler forms, $\omega_\text{WP}$ and $\omega_H$, by the formula \cite{R24,R25,R26,R27}
\be
\mathsf{Ric}_\text{WP}=-(m+3)\omega_\text{WP}+\omega_\text{H},
\ee
where $m\equiv\dim_\C\mspace{-1.5mu}\mathscr{X}$.

Let $\gamma(s)\subset \mathscr{X}$ be a geodesic for the WP metric parametrized by its arc-length.
If the geodesic arc between $\gamma(0)$ and $\gamma(L)$ is stable,  we have
the Myers inequality
\be\label{MMMMy}
\begin{split}
0\leq &\frac{\pi^2}{2L}(2m-1)-\int_0^L \sin^2\!\left(\frac{\pi s}{L}\right)\,\mathsf{Ric}_\text{WP}(\dot\gamma,\dot\gamma)\,ds\equiv\\
&\equiv \frac{\pi^2}{2L}(2m-1)+(m+3)\frac{L}{2}-\int_0^L \sin^2\!\left(\frac{\pi s}{L}\right)
\|\dot\gamma\|^2_\text{H}\,ds\leq \\
&\leq \frac{\pi^2}{2L}(2m-1)+(m+3)\frac{L}{2}-\int_a^b \sin^2\!\left(\frac{\pi s}{L}\right)
\|\dot\gamma\|^2_\text{H}\,ds
\end{split}
\ee
for all $0<a < b<L$. We use the same estimate of the integral as before. We write
$\boldsymbol{L}(\gamma;a,b)_\text{H}$ for the length of the arc of $\gamma$ from $\gamma(a)$ to $\gamma(b)$ as measured with the Hodge metric, and $\boldsymbol{L}(\gamma;a,b)_\text{WP}$
for its length measured with the WP metric: $\boldsymbol{L}(\gamma;a,b)_\text{WP}\equiv(b-a)L$. Eq.\eqref{MMMMy} then gives
\be
\boldsymbol{L}(\gamma;a,b)_\text{H}^2\leq \left(\frac{\pi^2}{2L}(2m-1)+(m+3)\frac{L}{2}\right)\int_a^b \frac{ds}{\sin^2(\pi s/L)}
\ee
for all $0<a < b<L$.
We may rewrite this inequality as an upper bound on the ratio of the lengths of the geodesic arc  between $a=\alpha L$ and $b=\beta L$ as measured by the Hodge metric, $\boldsymbol{L}(\gamma;\alpha L,\beta L)_\text{H}$, and by the Weil-Petersson one,
$\boldsymbol{L}(\gamma;\alpha L,\beta L)_\text{WP}$:
\begin{pro} Let $\gamma(s)\subset\mathscr{X}$ be a geodesic for the WP metric which is stable 
in the arc $0\leq s\leq L$.
For all $0<\alpha <\beta <1$ one has
\be\label{aaassss}
\left[\frac{\boldsymbol{L}(\gamma; \alpha L,\beta L)_\text{H}}{\boldsymbol{L}(\gamma; \alpha L,\beta L)_\text{WP}}\right]^{\!2}\leq \frac{\sin(\pi(\beta-\alpha))}{2\pi(\beta-\alpha)^2 \sin(\pi \alpha)\sin(\pi \beta)} \left(m+3+\frac{(2m-1)\pi^2}{L^2}\right).
\ee
\end{pro}
We may specialize this inequality in various limits, depending on the nature of the stable geodesic: finite, semi-infinite, or complete.
In the finite case $L$ can be taken at most equal to the geodesic distance between the singularities at the endpoints. With this choice,
if we send the second point $\gamma(\beta L)$ toward the singular point, i.e.\! take $\beta\to 1$,
the \textsc{rhs} diverges corresponding to the fact that the WP length of the curve is finite while the
Hodge length is infinite (this length is greater than the Hodge distance which goes to $+\infty$
since the Hodge metric is complete). 

The SUGRA geometry result corresponding to \textbf{Proposition \ref{xxxy}} for the QFT geometry is:

\begin{corl} Suppose the geodesic $\gamma(s)$ is complete and stable (this second assumption is superfluous if {\bf Sub-conjecture} holds). There are two universal positive constant $c_1$ and $c_2$, depending only on the dimension $m$ of $\mathscr{X}\!$\!, such that, for all pair of points $p,q\in\gamma$ the ratio of the lengths of the arc between them
measured with the Hodge and SW metrics satisfies
\be\label{yyytt}
\left[\frac{\boldsymbol{L}(\gamma;p,q)_\textrm{H}}{\boldsymbol{L}(\gamma;p,q)_\textrm{WP}}\right]^{\!2}\leq c_1+\frac{c_2}{\boldsymbol{L}(\gamma;p,q)^2_\textrm{WP}}
\ee
\end{corl}
\begin{proof} Set $L=2\,\boldsymbol{L}(\gamma;p,q)^2_\textrm{WP}$ and choose the parametrization so that
$p=\gamma(L/4)$ and $q=\gamma(3L/4)$. Then apply eq.\eqref{aaassss} to the stable geodesic arc between
$\gamma(0)$ and $\gamma(L)$ with $\alpha=1/4$ and $\beta=3/4$. We get \eqref{yyytt} with
\be
c_1=\frac{4}{\pi}(m+3),\qquad c_2=\pi(2m-1).
\ee 
One may  improve the constants.
\end{proof}
Thus, up to a universal change of normalization, asymptotically for large WP distances, the 
Hodge lengths of arcs of (stable) complete WP geodesics are not greater than their SW lengths.

\section*{Acknowledgments}
I have benefit of several discussions with Cumrun Vafa who prompted me to study special geometry from a Riemannian point of view. I thank him for sharing with me some of his deep insights on physics.

\appendix

\section{Geodesics near codimension-1 singularities}

As an illustration of several points raised in the main text, we study the asymptotic behavior of geodesics
near singularities in dimension one. Under \textbf{Condition} the finite cover $\cn$ is a smooth
variety of the form $\cn=\overline{\cn}\setminus D$ where $\overline{\cn}$ is smooth projective and $D$
a normal crossing divisor. We consider an irreducible component of $D$ which is at finite distance
in the special K\"ahler metric, and focus on a small neighborhood $U\subset\overline{\cn}$ of a \emph{generic} point on it (remaining far away from the crossing with other components where more intricate phenomena may happen). Restricting $U$, if necessary, we may
assume $U\cap \cn \simeq \Delta^*\times \Delta^{r-1}$ where $\Delta$ is a small disk and $\Delta^*$
is a small punctured disk. One can find special coordinates $\{a^1,a^2\dots,a^r\}$ in $\Delta^*\times \Delta^{r-1}$ ($a^1\in \Delta^\ast$, $a^i\in\Delta$ for $i\geq2$) 
so that the special metric locally reads \cite{R28}
\be
ds^2= -\frac{n}{2\pi}\log\mspace{-1mu} |a_1|\, da_1\, d\bar a_1+ \mathrm{Im}\!\left[\frac{\partial^2 \Psi(a)}{\partial a^i\,\partial a^j}\right] da^i d\bar a^j,\quad i,j=1,\dots,r,
\ee
where the function $\Psi(a)$ extends holomorphically over $U\simeq \Delta^r$ and $n$ is a positive integer. If we are only interested in the asymptotic behavior at a codimension-1 singularity, we
can neglect the second term and study the geodesics of the rescaled metric 
\be\label{ghdf}
ds^2=-\log r\big(dr^2+r^2 d\theta^2\big)
\ee
in the disk. This is a special K\"ahler geometry which does \emph{not} satisfy \textbf{Condition}.
Here we see what goes wrong when the \textbf{Condition} does not hold. In the unit disk 
the scalar curvature
\be
R=\frac{1}{2(-\log r)^3\,r^2}
\ee
has a genuine finite-distance singularity at $r=0$ but also a spurious one at $r=1$. The \textbf{Condition} gets rid of spurious finite-distance singularities, and in this sense replaces completeness which is ruled out for non-trivial special geometries. The period map is
\be
(r,\theta)\mapsto \frac{n}{2\pi}\theta+\frac{1}{2\pi i}\log r\in \mathscr{H}.
\ee

The geodesics Hamilton-Jacobi equation \cite{R29,R30} for the asymptotic metric \eqref{ghdf}
\be
\left(\frac{\partial S}{\partial r}\right)^{\!\!2}+\frac{1}{r^2}\!\left(\frac{\partial S}{\partial\theta}\right)^{\!\!2}+\log r=0
\ee 
can be reduced to quadrature by separation of variables
\be
S(r,\theta;J)=J\,\theta +\int_{r_-}^r \sqrt{-\log s-\frac{J^2}{s^2}}\;ds
\ee
The geodesics are parametrized by the `angular momentum' $J$ and take the form
\be\label{ggggeeedd}
\theta(r)=\theta_0+J\int_{r_-}^r \frac{ds}{s\sqrt{-s^2\log s-J^2}}
\ee
There are two kinds of geodesics: those with $J=0$ are radial lines at constant
angle, $\theta(r)=\theta_0$. They end up in the singularities at $r=0$ and at $r=1$.
Their length $L$, measured in the metric \eqref{ghdf}, is finite
\be
L\equiv\int_0^1 \sqrt{-\log r}\; dr =\frac{\sqrt{\pi}}{2}.
\ee
The geodesics with $J\neq0$ do not fall in the singularity at $r=0$ nor in the one at $r=1$: they are confined in the annulus $r_-<r<r_+$ where $0< r_-<r_+<1$ are the two zeros of the effective
potential $V_\text{eff}(r)=\log r+J^2/r^2$. These geodesics go around the singularity forever and so are \emph{complete.}
When pulled back to the cover $\mathscr{X}$ (i.e.\! by taking $\theta$ valued in $\R$ instead of $S^1$)
a complete geodesic will intersect an incomplete radial one exactly once since the angular velocity $\dot\theta$
is non-zero when $J\neq0$. We conclude 
\begin{fact} Non-complete geodesics in the asymptotic geometry \eqref{ghdf} are stable.
\end{fact}  
This \textbf{Fact} was already noticed in \cite{R15}. On the contrary the complete geodesics with $J\neq0$ cannot be stable by \textbf{Proposition \ref{xxxy}}. The complete geodesics are parametrized by the angular momenta $J>0$ such that $V_\text{eff}(r)$ has two real zeros in the interval $(0,1)$. When these zeros coincide
$r_-=r_+=r_0$ we have a circular geodesic and the Gauss-Bonnet theorem yields
\be\label{gggsss}
1=\frac{1}{2\pi}\int_{r\leq r_0} \!\!\!\mathsf{Ric} \equiv\frac{1}{2}\int_0^{r_0} \frac{dr}{(\log r)^2\,r}=\frac{1}{2(-\log r_0)}
\ee
To describe a complete geodesic $J$ must be such that $\min_r V_\text{eff}(r)\leq 0$ i.e.
\be
0<J\leq J_c\equiv \frac{1}{\sqrt{2e}}
\ee
the minimum being attained at $r=\sqrt{2} J$. See figure \ref{figure}.
When the inequality is saturated $r_-=r_+$ and we have a circular orbit with
radius $r_0=\frac{1}{\sqrt{e}}$ in agreement with the Gauss-Bonnet formula \eqref{gggsss}. Let us check that the complete geodesic are not stable. It is enough to check the property for the circular one; then the statement is a particular case of Bertrand theorem in Mechanics \cite{R30}.
We consider an angular momentum just below the
 critical value
 \be
J^2=\frac{1}{2e}(1-2\epsilon^2)\qquad \epsilon\lll1.
\ee
We parametrize the radius $r$ in terms of a new coordinate $u$
\be
r=\frac{1-\epsilon\, u}{\sqrt{e}}
\ee
 so that
 \be
 -\log r-\frac{J^2}{r^2}=\epsilon^2(1-u^2) +O(\epsilon^3)
 \ee
 Replacing in eq.\eqref{ggggeeedd} we get
 \be
 \theta=\theta_0+\frac{1}{\sqrt{2}}\int_{u_0}^{u(\theta)}\frac{du}{\sqrt{1-u^2}}+O(\epsilon)
 \ee
  and performing the elementary integral, 
 we get for $\epsilon\approx0$
 \be
 r=\frac{1}{\sqrt{e}}\big(1-\epsilon \sin[\sqrt{2}(\theta-\theta_0)]+O(\epsilon^2)\big)
 \ee
 which is just the small oscillations around the circular orbit $r=1/\sqrt{e}$ corresponding to
 a non-zero Jacobi field with zeros at $\theta=\theta_0+\pi k/\sqrt{2}$ ($k\in\Z$).

\begin{figure}
$$
\includegraphics[width=0.34\textwidth]{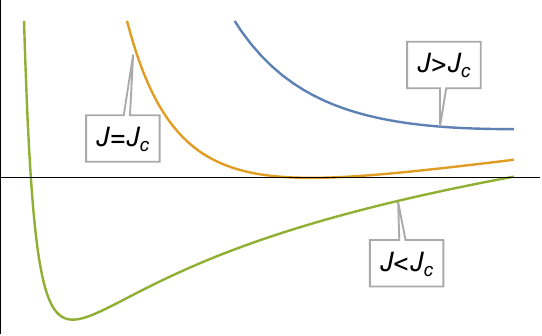}
$$
\caption{\label{figure} The effective radial potential $V_\text{eff}(t)$ for various values of $J$}
\end{figure}


\begin{thebibliography}{183}
  \begin{footnotesize}
  

  
  \bibitem{R1}
  N. Seiberg and E. Witten, Electric-magnetic duality, monopole condensation, and confinement in N=2 supersymmetric Yang-Mills theory, Nucl. Phys. B 426, 19 (1994) Erratum: Nucl. Phys. B 430, 485 (1994), arXiv:hep-th/9407087.
  
  \bibitem{R2}
N. Seiberg and E. Witten, Monopoles, duality and chiral symmetry breaking in N=2 super- symmetric QCD, Nucl. Phys. B 431, 484 (1994), arXiv:hep-th/9408099.
  
  \bibitem{R3}
R. Donagi and E. Witten, Supersymmetric Yang-Mills theory and integrable systems,
Nucl. Phys. B 460, 299 (1996), arXiv:hep-th/9510101.
  
  \bibitem{R4}
R. Donagi, Seiberg-Witten integrable systems, in \textit{Surveys in Differential Geometry}, Vol. IV (1998) pp. 83-129.
  
  \bibitem{R5}
  D.S. Freed, Special K\"ahler manifolds, Commun. Math. Phys. 203 (1999) 31-52
  
  \bibitem{R6}
M.~Caorsi and S.~Cecotti,
Geometric classification of 4d $\mathcal{N}=2$ SCFTs,
JHEP \textbf{07}, 138 (2018),
arXiv:1801.04542.
  
  \bibitem{R7}
 S.~Cecotti,
Direct and inverse problems in special geometry,
arXiv:2312.02536 [hep-th].

\bibitem{R7.5}
S.~Cecotti, M.~Del Zotto, M.~Martone and R.~Moscrop,
The characteristic dimension of four-dimensional ${\mathcal {N}}$~=~2 SCFTs,
Commun. Math. Phys. \textbf{400} (2023) no.1, 519-540;
arXiv:2108.10884 [hep-th].
  
  \bibitem{R8}
S. Cecotti, \textit{Supersymmetric Field Theories. Geometric Structures and Dualities,}
Cambridge University Press, 2015
  
  \bibitem{R9}
 S. Cecotti, N = 2 Supergravity, Type IIB superstrings and Algebraic Geometry, Commun. Math. Phys. 131, 517 (1990).
  
  \bibitem{R10}
A. Strominger, Special geometry, Comm. Math. Phy. 133 (1990)163-180.


  
  \bibitem{R11}
S.~Cecotti,
Special geometry and the swampland,
JHEP \textbf{09} (2020), 147,
arXiv:2004.06929 [hep-th].
  
  \bibitem{R12}
R. Bryant and P. Griffiths, Some observations on the infinitesimal period relations for regular threefolds with trivial canonical bundle, in M. Artin and J. Tate (Eds.) \textit{Arithmetic and Geometry. Papers dedicated to I.R. Shafarevich,} vol. II Birk\"auser (1983) pages 77-102.


  
  \bibitem{R13}
S.~Cecotti,
The Weil correspondence and universal special geometry,
JHEP \textbf{07} (2024), 020;
arXiv:2404.16316 [hep-th].
  
  \bibitem{R14}
M.~Caorsi and S.~Cecotti,
Special Arithmetic of Flavor,
JHEP \textbf{08} (2018), 057;
arXiv:1803.00531 [hep-th].

  
  \bibitem{R15}
S.~Raman and C.~Vafa,
Swampland and the geometry of marked moduli spaces,
arXiv:2405.11611 [hep-th].

\bibitem{R15a}
 C. Vafa, The string landscape and the swampland, Oct., 2005. 10.48550/arXiv.hep-th/0509212.

\bibitem{R15b}
 T.D. Brennan, F. Carta and C. Vafa, The string landscape, the swampland, and the missing corner, June, 2018. 10.48550/arXiv.1711.00864.
 
 \bibitem{R15c}
M. van Beest, J. Calder\'on-Infante, D. Mirfendereski and I. Valenzuela, Lectures on the Swampland Program in String Compactifications, Physics Reports 989 (2022) 1

\bibitem{R15d}
N.B. Agmon, A. Bedroya, M.J. Kang and C. Vafa, Lectures on the string landscape and the Swampland, Mar., 2023. 10.48550/arXiv.2212.06187.

\bibitem{R15e}
E. Palti, The Swampland: Introduction and review, Fortschritte der Physik 67 (2019) 1900037. 

\bibitem{R15f}
M. Gra\~na and A. Herr\'aez, The swampland conjectures: A bridge from Quantum Gravity to
particle physics, July, 2021. 10.48550/arXiv.2107.00087.
  
  \bibitem{R16}
K. Liu, X. Sun, X. Yang and S.-T. Yau, Curvatures of moduli space of curves and applications, \texttt{arXiv:1312.6932.}


\bibitem{R21}
J.-P. Demailly, \textit{Complex analytic and algebraic geometry.} Book online https://www-fourier.ujf-grenoble.fr/~demailly/manuscripts/agbook.pdf


\bibitem{R16.5}
J. Ells and J.H. Sampson, Mappings of Riemannian manifolds,
Am. J. of Math. \textbf{86} (1964) 109-160. 
  
  
  \bibitem{R17}
M. Green, P. Griffiths and M. Kerr, \textit{Mumford-Tate Groups and Domains: Their Geometry and Arithmetic,} Annals of Mathematics Studies, PUP (2012).

 \bibitem{R18}
J. Carlson, S. M\"uller-Stach, C. Peters, \textit{Period Mappings and Period Domains, Second Edition}, Cambridge studies in advanced mathematics 168, CUP (2017).
  
\bibitem{Bor}
A. Borel, Some metric properties of arithmetic quotients of symmetric spaces and an extension theorem,
J. Diff. Geom. \textbf{6} (1972) 543-560.

\bibitem{BaiB}
W.L. Baily and A. Borel, Compactification of arithmetic quotients of bounded
symmetric domains, Ann. of Math. \textbf{84} (1966) 442-528.

\bibitem{Hart}
R. Hartshorne, \textit{Algebraic Geometry,} Graduate Texts in Mathematics 52, Springer 1977.
  
  \bibitem{R19}
J. J\"ost, \textit{Riemannian Geometry and Geometric Analysis,} 7th Edition, 
Springer 2007.
  
  \bibitem{R20}
K. Corlette, Flat G-bundles with canonical metrics, J. Diff. Geom. 28 (1988) 361-382.

\bibitem{yano}
K. Yano, \textit{The theory of the Lie derivative and its applications,} Nabu Press (2011).
  
%
  
  \bibitem{R23}
P.A. Griffiths, Periods of integrals on algebraic manifolds, III (Some global differential-geometric properties of the period mapping), Publ. IHES 38 (1970) 125-180.
  
  
  \bibitem{R24}
S.~Cecotti and C.~Vafa,
Ising model and N=2 supersymmetric theories,
Commun. Math. Phys. \textbf{157} (1993), 139-178;
arXiv:hep-th/9209085 [hep-th].
  
  \bibitem{R25}
  Z. Lu, On the Hodge metric of the universal deformation space of Calabi-Yau threefolds, 
  J. Geom. Anal. \textbf{11} (2001) pp 103-118, arXiv:math/0505582.

  
  \bibitem{R26}
  Z. Lu, On the curvature tensor of the Hodge metric of moduli space of polarized Calabi-Yau threefolds, arXiv:math/0505583
  
  \bibitem{R27}
S.~Cecotti,
Moduli spaces of Calabi-Yau $d$-folds as gravitational-chiral instantons,
JHEP \textbf{12} (2020), 008;
arXiv:2007.09992 [hep-th].
  
  \bibitem{R27b}
M.~Bershadsky, S.~Cecotti, H.~Ooguri and C.~Vafa,
Holomorphic anomalies in topological field theories,
Nucl. Phys. B \textbf{405} (1993), 279-304;
arXiv:hep-th/9302103 [hep-th].
  
  \bibitem{R28}
J.-M. Hwang and K. Oguiso, Local structure of principally polarized stable Lagrangian fibrations, arXiv:1007.2043.
  
  \bibitem{R29}
A. Fasano and S. Marmi, \textit{Analytic Mechanics. An introduction,}
Oxford Graduate Texts OUP, 2006.
  
  \bibitem{R30}
S. Cecotti, \textit{Analytic Mechanics. A concise textbook,}
UNITEXT for Physics, Springer, 2024.
  
  \end{footnotesize}
  \end{thebibliography}
\end{document}